\documentclass[10pt,twocolumn]{article} 

\usepackage[utf8]{inputenc}

\usepackage[british]{babel} 

\usepackage[top=2.54cm,bottom=2.54cm,left=3.17cm,right=3.17cm]{geometry} 

\usepackage{hyperref} 

\usepackage{fancyhdr} 

\usepackage{emptypage} 

\usepackage{enumitem}
\usepackage{caption}

\usepackage[fleqn]{amsmath}
\usepackage{amssymb}
\usepackage{amsfonts}
\usepackage{mathtools}
\usepackage{array}
\usepackage{siunitx}
\usepackage{algorithmic}
\usepackage{float}

\usepackage{setspace}

\usepackage{caption} 
\usepackage{graphicx}
\usepackage{graphics}
\usepackage{subfigure}

\usepackage{booktabs} 
\usepackage{multirow}

\usepackage[square,comma,numbers,sort&compress]{natbib}

\usepackage{xcolor}

\usepackage[noblocks]{authblk}

\usepackage[toc,page]{appendix}

\usepackage{ulem}

\numberwithin{table}{section} 
\numberwithin{equation}{section} 
\graphicspath{{../Figs/}} 
\DeclareGraphicsExtensions{.eps,.mps,.pdf,.jpg,.png} 
\definecolor{IgorGreen}{RGB}{76,153,0}
\definecolor{IgorBlue}{RGB}{0,0,204}
\definecolor{IgorGreen1}{RGB}{0,153,77}
\definecolor{IgorPurple}{RGB}{103,0,204}

\newtheorem{theorem}{Theorem}

\newenvironment{proof}[1][Proof]{\noindent\textbf{#1.} }{\ \rule{0.5em}{0.5em}}
\begin{document}
\title{SIS Epidemic Spreading with Heterogeneous Infection Rates}

\author[1]{Bo Qu}      
\author[1]{Huijuan Wang}
\date{}
\affil[1]{Delft University of Technology, Delft, Netherlands}
\maketitle

\begin{abstract}
In this work, we aim to understand the influence of the heterogeneity of infection rates on the Susceptible-Infected-Susceptible (SIS) epidemic spreading. Employing the classic SIS model as the benchmark, we study the influence of the independently identically distributed infection rates on the average fraction of infected nodes in the metastable state. The log-normal, gamma and a newly designed distributions are considered for infection rates. We find that, when the recovery rate is small, i.e.\ the epidemic spreads out in both homogeneous and heterogeneous cases: 1) the heterogeneity of infection rates on average retards the virus spreading, and 2) a larger even-order moment of the infection rates leads to a smaller average fraction of infected nodes, but the odd-order moments contribute in the opposite way; when the recovery rate is large, i.e.\ the epidemic may die out or infect a small fraction of the population, the heterogeneity of infection rates may enhance the probability that the epidemic spreads out. Finally, we verify our conclusions via real-world networks with their heterogeneous infection rates. Our results suggest that, in reality the epidemic spread may not be so severe as the classic SIS model indicates, but to eliminate the epidemic is probably more difficult.  
\end{abstract}

\section{Introduction}
{T}{he} studies on contagion processes in networks are strongly motivated and justified by the anticipated outbreaks of epidemic diseases in a population and non-stop threats of cyber security in computer networks \cite{albert2000error,kephart1991directed,garetto2003modeling,ganesh2005effect,majdandzic2014spontaneous}. The Susceptible-infected-susceptible (SIS) model \cite{daley2001epidemic,van2011n,cator2013susceptible,li2012susceptible,boccara1993critical,shi2008sis,wang2013effect} is one of the most widely used models to describe such processes. In the continuous-time Markovian SIS model, a node is either infected or susceptible at any time $t$. Each infected node infects each of its susceptible neighbors with an infection rate $\beta$. The infected node can be recovered with a recovery rate $\delta$. Both infection and recovery processes are independent Poisson processes. The average fraction $y_\infty$ of the infected nodes in the metastable state, ranging in $[0,1]$, indicates how severe the influence of the virus is: the larger $y_\infty$ is, the more severely the network is infected. 

The classic SIS model assumes that the infection rate $\beta$ is the same for all infected-susceptible node pairs and so is the recovery rate $\delta$ for all nodes. Most studies are focusing on the relationship between the effective infection rate $\tau$  and the average fraction $y_\infty$ of infected nodes or the epidemic threshold in the virus contamination process with homogeneous infection (recovery) rates. However, in reality, neither the contact frequency \cite{Fratiglioni20001315} between a pair of individuals in social networks nor the connecting frequency between a pair of nodes in computer networks is constant. Infection rates can be different from pairs to pairs, thus heterogeneous. Many studies on real diseases, such as SARS \cite{WANGWenBin:2143} and Plasmodium falciparum infection \cite{smith2005entomological} also reveal the heterogeneity of infection rates. Furthermore, Smith et al.\ \cite{smith2005entomological} suggest that the distribution of infection rates in different populations may be varied as well, and Wang et al.\ \cite{WANGWenBin:2143} find that infection rates with the log-normal distribution fit best the data of SARS in 2003 by applying their model.  

In this paper, we explore the effect of heterogeneous infection rates on the average fraction $y_\infty$ of infected nodes in a systematic way. We propose a SIS model, in a network with $N$ nodes, with the homogeneous recovery rate $\delta$ but heterogeneous infection rates $\beta_{ij}$ ($=\beta_{ji}$, $i=1,2,...,N$, $j=1,2,...,N$ and $i\neq j$) between node $i$ and node $j$. Similar to the classic homogeneous SIS, our SIS model with heterogeneous infection rates is as well a Markovian process where the time for an infected node $i$ to infect each of its susceptible neighbors $j$ is an independent exponential random variable with average $\beta_{ij}^{-1}$. The homogeneous SIS model has the same infection rate $\beta$ for all node pairs whereas all the infection rates in our heterogeneous SIS are independent and identically distributed (i.i.d.) random variables. We study how the distribution of infection rates influences the average fraction $y_\infty$ of infected nodes in the metastable state. 

A few recent papers \cite{qu2014heterogeneous,preciado2013optimal,preciado2014optimal, fu2008epidemic,buono2013slow,yang2012epidemic} have taken into account either the heterogeneous infection or recovery rates. In \cite{qu2014heterogeneous}, we explored the influence of degree-based recovery rates on the average fraction of infected nodes in the metastable state. Preciado et al.\ \cite{preciado2013optimal, preciado2014optimal} discussed how to choose the infection and recovery rates from given discrete sets to let the virus die out. Fu et al.\ \cite{fu2008epidemic} studied the epidemic threshold when the infection rates depend on the node degrees and Buono et al.\ \cite{buono2013slow} considered a specific distribution of infection rates and observed slow epidemic extinction phenomenon. Yang and Zhou \cite{yang2012epidemic} gave an edge-based mean-field solution of the epidemic threshold in regular networks (the degrees of all nodes are the same) with i.i.d.\ heterogeneous infection rates (following uniform or power-law distribution). 

In this paper, we explore the influence of heterogeneous infection rates on the epidemic spreading. In practice, the number of new infections in a period of time can be used to estimate the infection rate, for example, \cite{riley2003transmission} counts the number of infected people per time interval (daily, weekly, etc.) to indicate the infection rate; \cite{cook2007estimation}, illustrating a strategy to estimate the time-varying transmission rates for the spread of infection, also takes into account the daily distribution of new infections. Besides the number of new infections, the interacting frequencies between two neighboring nodes have also been employed to estimate the infection rate, for example, the infection rate has been considered to be proportional to the interacting frequency. The average infection rate obtained in both scenarios has been used as the infection rate in the homogeneous epidemic model. Our work points out how such assumption of homogeneous rates would differ from real-world heterogeneous infection rates with respect to their influence on the fraction of infected population. We consider several representative distributions with the same mean but higher moments tunable, since the influence of the mean\footnote{The infection rate between any pair of nodes equals to the mean in the homogeneous SIS model.} has been widely studied in the homogeneous SIS model \cite{pastor2001epidemic,pastor2005epidemics, pastor2014epidemic,guo2013epidemic,li2013epidemic}. To our best knowledge, our work is the first to discuss the influence of higher moments of the infection rate distribution in epidemic models. 

\section{SIS model with heterogeneous infection rates}
In this section, we introduce the classic SIS model, basic network models, the heterogeneous infection rates and the simulation settings of the SIS model with heterogeneous infection rates on a network.  
\subsection{The classic SIS model}
In the continuous-time Markovian SIS model on a network with $N$ nodes, the state of a node at any time $t$ is a Bernoulli random variable, where $X_i(t)=0$ represents that node $i$ is susceptible and $X_i(t)=1$ that node $i$ is infected. Each infected node infects each of its susceptible neighbors with an infection rate $\beta$. The infected node can recover with a recovery rate $\delta$. Both infection and recovery processes are independent Poisson processes. The ratio $\tau=\frac{\beta}{\delta}$ is called the effective infection rate. For each effective infection rate $\tau$, the infection process dies out in any finite network after a long enough time, and the corresponding steady state is the absorbing state: \textit{i.e.}\ the overall healthy state. However, if the effective infection rate $\tau$ is larger than the epidemic threshold $\tau_c$, the epidemic spreads out and there is a non-trivial metastable state, where the average fraction $y_\infty$ of infected nodes is non-zero and stable during a long time \cite{van2013homogeneous}. The average fraction $y_\infty$ of infected nodes indicates the severity of the overall infection.

\subsection{Network models}
Among various network models, Erd\"os-R\'enyi (ER) model\cite{erdds1959random} is one of the most widely-used and well-studied models. In an ER random network with $N$ nodes, each pair of nodes are connected with probability $p$ independent from every other pair, thus the distribution of the degree of a random node is binomial: $Pr[D=k]=\binom{N-1}{k}p^k(1-p)^{N-1-k}$ and the average degree $E[D]=(N-1)p$. For a large $N$ and constant average degree, the degree distribution is Poisson: $Pr[D=k]=e^{-Np}(Np) ^{k}/k!$. 

Besides the ER model, the network model with a scale-free degree distribution (SF model) has always been used to describe real-world networks such as the Internet\cite{caldarelli2000fractal} and World Wide Web\cite{albert1999internet}. The degree distribution of SF
networks is given by $Pr[D=k]\scriptsize{\sim}k^{-\lambda}, k\in[d_{\rm
    min},d_{\rm max}]$, where $d_{\rm min}$ is the smallest degree,
$d_{\rm max}$ is the degree cutoff, and $\lambda$ is the exponent
characterizing the broadness of the distribution
\cite{barabasi1999emergence}. In real-word networks, the exponent $\lambda$ is usually in the range $[2,3]$, thus we confine the exponent $\lambda=2.5$ in this paper. We further employ the smallest degree $d_{\rm min}=2$, the natural degree cutoff $d_{\rm max}=\lfloor N^{1/(\lambda-1)}\rfloor$ \cite{PhysRevLett.85.4626} , and the size $N=10^4$. Hence, the average degree is approximately $4$. As the comparison, we consider the ER networks with the size $N=10^4$ and the average degree $E[D]=4$.

\subsection{Heterogeneous infection rates}
\label{Chap:Hbeta}
In this subsection, we introduce three distributions of the heterogeneous infection rates. We aim to explore how the heterogeneous infection rates influence the spread of SIS epidemics, particularly we study the relationship between the variance\footnote{The variance of a random variable is the second central moment.} (and even higher moments) of the heterogeneous infection rates and the average fraction $y_\infty$ of infected nodes. Hence, we would like to choose infection-rate distributions systematically such that they cover a broad range of distributions including those observed in real-world and importantly their higher order moments, at least the variances are tunable when their means are fixed. 

The $nth$ moment $m_n$ of a distribution with the probability density function (PDF) $f_B(\beta)$ is $m_n=\int\limits_{-\infty}^{+\infty}\beta^nf_B(\beta)d\beta$. Thus, the first moment $m_1$ is just the mean and the relationship between the second moment $m_2$ and variance $Var[B]$ is $Var[B]=m_2-m_1^2$, where the random variable $B$ is the infection rate of a link. To eliminate the influence of the mean $m_1$, we further define the $nth$ normalized moment $\nu_n=\frac{m_n}{m_1^n}$, then $\nu_1=1$ and the normalized variance $v=\nu_2-1$. 

We choose two asymmetric distributions: the log-normal and gamma distribution, of which we can keep the means unchanged and tune the variances in a large range. The log-normal distribution \cite{van2006performance} $B\sim Log\textrm{-}\mathcal{N}(\beta;\mu,\sigma)$, of which the PDF is, for $\beta>0$
\[
f_B(\beta; \mu, \sigma)=\frac{1}{\beta\sigma \sqrt{2\pi}}exp\left(-\frac{(\ln \beta-\mu)^2}{(2\sigma^2)}\right)
\]
and the $nth$ normalized moment is $\nu_n=exp(\frac{(n^2-n)\sigma^2}{2})$, has a power-law tail for a large range of $\beta$ provided $\sigma$ is sufficiently large. The log-normal distribution has as well been widely observed in real-world, where the interaction frequency between nodes is usually considered as the infection rate between those nodes. One example is the infection rates of the co-author network, as illustrated in Fig.~\ref{fig:52}, Section 5. Moreover, Wang et al.~\cite{WANGWenBin:2143} find that by employing the log-normal distributed infection rates, their epidemic model can accurately fit the infection data of 2003 SARS.

The gamma distribution $B\sim \Gamma(\beta;k,\theta)$, of which the PDF is, for $\beta>0$
\[
f_B(\beta;k,\theta)=exp(-\frac\beta\theta)\frac{\beta^{k-1}}{\theta^{k}\Gamma(k)}
\] 
($\Gamma(k)=\int\limits_{0}^{\infty}t^{k-1}e^{-t}dt$) and the $nth$ normalized moment is $\prod\limits_{i=0}^{n-1}(1+ik^{-1})$, has a lighter tail than the log-normal distribution. The Airline network, as demonstrated in Fig.~\ref{fig:51}, has an exponentially distributed infection rates, which corresponds to the Gamma distribution when $k=1$. 

In order to take into account symmetrically distributed infection rates as well, we design a variance-tunable and symmetric distribution other than the two asymmetric distributions above. We call it the symmetric polynomial (SP) distribution $B\sim SP(\beta;a,b)$, whose PDF is
\[
f_B(\beta;a,b)=\frac{b(a+1)}{2}|\beta-1|^a
\]
where $\beta\in [1-\frac{1}{\sqrt{b}},1+\frac{1}{\sqrt{b}}])$ and, $a=1$ and $b\in [1,+\infty)$ or $b=1$ and $a\in [1,+\infty)$. The mean of the distribution is $1$, the variance is $\frac{a+1}{b(a+3)}$.
Compared to the commonly-used uniform distribution (also symmetric and variance-tunable) with the same mean, the SP distribution can be tuned in a larger range of the variance.  

\subsection{The simulations}
In order to study the effect of the variance of the heterogeneous infection rates on the virus spread, we perform simulations to obtain the fraction $y_\infty$ of infected nodes as a function of the normalized variance $v$ of infection rates on both ER and SF networks. We find that, for commonly used 2-parameter distributions (such as the uniform distribution, log-normal distribution, gamma distribution, etc.), the scaling on the mean of infection rates can be eliminated by the same scaling on the recovery rate if we keep the normalized variance $v$ unchanged. This conclusion is also consistent with the fact that only the effective infection rate $\frac\beta\delta$ matters for the epidemic spreading, but not the infection rate $\beta$ in the homogeneous SIS model. Hence, without loss of generality,  we set the mean $m_1$ of the infection rates to $1$, thus all the normalized moments $\nu_n$ equal to the unnormalized ones $m_n$. Instead of performing discrete-time simulations, we further develop a continuous-time simulator, which was firstly proposed by van de Bovenkamp and described in detail in \cite{li2012susceptible}. A discrete-time simulation could well approximate a 
continuous process if a small time bin to sample the continuous process is selected so that within each time bin, no multiple events occur. A heterogeneous SIS model allows different as well large infection or recovery rates, which requires even smaller time bin size and challenges the precision of a discrete-time simulation. Hence, we implement the precise continuous-time simulations. 

\section{Small recovery rates} 
\label{smallRecov}
In this work, the average of the heterogeneous infection rates and the homogeneous infect rate are the same. Since the recovery rate $\delta$ plays the key role in the epidemic spreading, we discuss our results according to different ranges of the recovery rates. In this section, we introduce our main results about how the heterogeneous infection rates influence the contagion processes of epidemic, when the recovery rates are small such that the epidemic spreads out in both homogeneous and heterogeneous cases. In the next section, we focus on large recovery rates -- the homogeneous effective infection rate $\tau$ is close to the epidemic threshold $\tau_c$ in the classic model, where the epidemic with homogeneous infection rates may die out.       
\subsection{The observations}
\label{sec:observation}
We first show the simulation results when the variance of the infection rates is smaller than $1$, since the variance of a non-negative and symmetric distribution cannot be larger than the square of its mean\footnote{For any random variable $B$ following a non-negative and symmetric distribution $f_B(\beta)$ with mean $m_1$, the smallest and largest value that $B$ can reach is $0$ and $2m_1$ respectively, so the largest variance, which equals to $m_1^2$, can be reached when $Pr[B=0]=Pr[B=2m_1]=0.5$. }, thus $1$ in this paper. 
\begin{figure}[htbp]
\centering
\includegraphics[scale=.3]{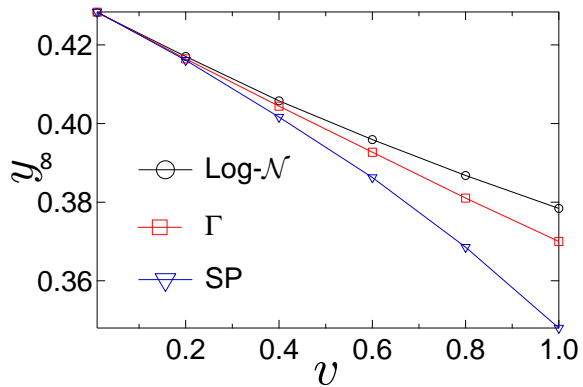}
\caption{The average fraction $y_\infty$ of infected nodes as a function of the normalized variance $v$ of infection rates for log-normal ($\circ$), gamma ($\textcolor{red}{\square}$), and SP (($\textcolor{IgorBlue}{\triangledown}$) infection-rate distributions respectively, and the recovery rate $\delta=2$. We consider ER networks with average degree $E[D]=4$ and network size $N=10^4$. The results are averaged over $1000$ realizations.}
\label{fig:HSIS-1}
\end{figure}
 
In Fig.~\ref{fig:HSIS-1}, we find that the average fraction $y_\infty$ of infected nodes decreases as the variance $v$ of the infection rates increases, no matter which distribution the infection rates follow. Moreover, the comparison of the decay of the three curves in Fig.~\ref{fig:HSIS-1} also suggests that, the smaller the third moment\footnote{The third moment of the log-normal, gamma and SP distribution is $(v+1)^3$, $(v+1)(2v+1)$ and $3v+1$ respectively.} of the infection rate distribution is, the faster $y_\infty$ decays as the variance increases. 

When the variance $v$ is larger than $1$, the infection rates cannot be symmetrically distributed. We thus discuss only the log-normal and gamma distributions which are representative among the heavy-tailed distributions and widely used in the real-world analysis.

\begin{figure}[!t]
\centering
\includegraphics[scale=.3]{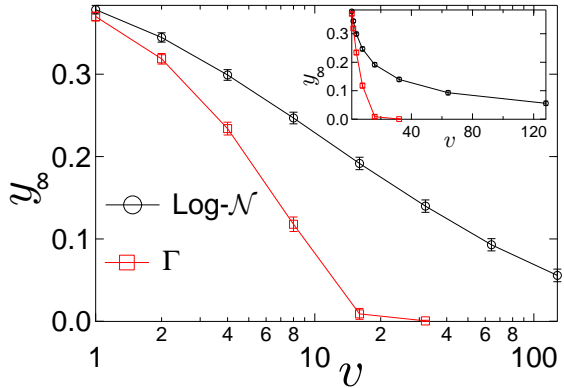}
\caption{The average fraction $y_{\infty}$ of infected nodes as a function of the variance $v$ of infection rates following different distributions: log-normal ($\circ$) and gamma ($\textcolor{red}{\square}$), and the recovery rate $\delta=2$. The simulations are on ER networks with average degree $E[D]=4$ and network size $N=10^4$. The results are averaged over $1000$ realizations, and the error bars are the standard deviations of the results in different realizations. The inset is the same as the main graph, but in a linear-linear scale.}
\label{fig:HSIS-2}
\end{figure} 
 
In Fig.~\ref{fig:HSIS-2}, we observe the same as in Fig.~\ref{fig:HSIS-1}. Moreover, we find that the average fraction $y_{\infty}$ of infected nodes decays much faster when infection rates follow gamma distributions than log-normal distributions. 

Here we only show the simulation results on ER random networks with $10^4$ nodes and average degree $E[D]=4$, because simulation results on SF networks lead to the same observations as illustrated in the Appendix. Moreover, though not shown in this paper, we have also done the simulations with various values of the recovery rate, such as $\delta=0.1$, $0.2$, $1$, etc., for both ER and SF networks and the conclusions are consistent. 

\subsection{The influence of the moments of the infection rates}
To explain our observations, we consider a susceptible node and an infected node interconnected by a link. The probability $\rho(T)$ that the infected node infects the susceptible neighbor in an arbitrary period $T$, is $\rho (T)=\int\limits_{0}^{+\infty}f_B(\beta)F(T;\beta)\,\mathrm{d}\beta$, where $f_B(\beta)$ is the PDF of the infection rate, and $F(T;\beta)$ is the probability that infection occurs between the neighboring infected and susceptible node pair within the time interval $T$ when the infection rate is $\beta$. Since the infection between any infected and susceptible node pair is an independent Poisson process, the time for an infected node to infect a susceptible neighbor is an exponential variable, \textit{i.e.}\ $F(T;\beta)=1-e^{-T\beta}$. We consider further the classic homogeneous SIS model, whose infection rate is equal to the average infection rate $E[B]$ in our heterogeneous SIS model. The counterpart of $\rho(T)$ in the homogeneous SIS model is, then, $\rho^*(T)=F(T;E[B])$. 

\begin{theorem}
If $f_B(\beta)$ is the probability density function of a non-negative continuous random variable $B$, and $F(T;\beta)$ is the distribution function of an exponential random variable with the rate parameter $\beta$, then for any $T>0$, we have
\[
\int\limits_{0}^{\infty}f_B(\beta)F(T;\beta)\,\mathrm{d}\beta\leq F(T;E[B])
\]    
\end{theorem}
\begin{proof}
\[
\begin{aligned}
&\int\limits_{0}^{\infty}f_B(\beta)F(T;\beta)\,\mathrm{d}\beta\\
=& 1-\int\limits_{0}^{\infty}f_B(\beta)e^{-T\beta} \,\mathrm{d}\beta \\
=& 1-E[e^{-T B}]
\end{aligned}
\]
Since the exponential function is convex, Jensen's inequality \cite{van2006performance} tells us that
\[
E[e^{-T B}]\geq e^{-T E[B]}
\]
Hence, 
\[
\int\limits_{0}^{\infty}f_B(\beta)F(T;\beta)\,\mathrm{d}\beta\leq 1-e^{-T E[B]} =F(T;E[B])
\]
\end{proof}

Theorem 1, that proves $\rho(T)\leq\rho^*(T)$, tells us that if the infection rate in the classic homogeneous SIS model and the average infection rate in heterogeneous model are the same, then in the same period of time an infection event is more likely to happen in the classic SIS model.

We define $\chi (T)=\rho^*(T)-\rho(T)$ as the difference in infection probability within an arbitrary time interval $T$ between the SIS model with homogeneous and heterogeneous infection rates. 
\begin{equation}
\label{Equ1}
\begin{aligned}
\chi (T ) = & E[e^{-T B}]-e^{-T E[B]} \\
=& \sum_{n=0}^{\infty }\frac{(-T )^{n}\left( E[B^{n}]-\left( E[B]\right)
^{n}\right) }{n!} \\
=& \sum_{n=0}^{\infty }\frac{\left( m_{n}-m_{1}^{n}\right) (-T )^{n}}{n!}
\\
=& \sum_{n=0}^{\infty }(\nu _{2n}-1) \frac{(T m_{1})^{2n}}{%
\left( 2n\right) !}-\sum_{n=0}^{\infty }(\nu _{2n+1}-1) \frac{%
(T m_{1})^{2n+1}}{(2n+1)!}  
\end{aligned}
\end{equation}

Note that the first step in (\ref{Equ1}) is valid only if the sum $\sum_{n=0}^{\infty }\frac{(-T
)^{n}E[B^{n}]}{n!}$ converges. The general log-normal distribution over an infinite range does not satisfy
this condition. However, the infection rates of real-world systems are
finite. Theorem 2 states that any realistic distribution of the infection rates
within a finite range satisfies this convergence condition.

\begin{theorem}
For any non-negative random variable $B$ distributed in a finite range $[0,b]$ and any finite $T$, the sum \[\sum_{n=0}^{\infty }\frac{(-T)^{n}E[B^{n}]}{n!}\leq 2 e^{T b}\] thus converges.
\end{theorem}

\begin{proof}
\[
\begin{aligned}
E[B^{n}] &=\int_{0}^{b}\beta ^{n}f_{B}(\beta )d\beta \\
&=\left.\beta ^{n}\int_{0}^{\beta }f_{B}(\beta )d\beta \right\vert
_{0}^{b}-\int_{0}^{b}\int_{0}^{\beta }f_{B}(\beta )d\beta d\beta ^{n} \\ 
&=\left. \beta ^{n}F_{B}(\beta )\right\vert
_{0}^{b}-\int_{0}^{b}F_{B}(\beta )d\beta ^{n}
\end{aligned}
\]%
Since
\[
F_{B}(\beta )=\int_{0}^{\beta }f_{B}(\beta )d\beta \leq 1
\] we have 
\[E[B^{n}]\leq b^{n}+\left\vert \int_{0}^{b}F_{B}(\beta )d\beta
^{n}\right\vert \leq 2b^{n}.
\]
Hence,%
\[
\begin{aligned}
&\left\vert \sum_{n=0}^{\infty }\frac{(-T )^{n}E[B^{n}]}{n!}\right\vert \\
\leq &\sum_{n=0}^{\infty }\frac{|(-T )^{n}||E[B^{n}]|}{n!} \\
\leq & 2\sum_{n=0}^{\infty}\frac{T ^nB^n}{n!}\\
= & 2 e^{T b}
\end{aligned}
\]
which illustrates the convergence of $\sum_{n=0}^{\infty }\frac{(-T)^{n}E[B^{n}]}{n!}$ for any $T$.\\
\end{proof}

Theorem 1 and (\ref{Equ1}) explore only on the local effect: the epidemic spreads on average faster along a link in the heterogeneous case than the homogeneous case. However, if the infection probabilities of all the nodes are similar and the state of the each node (infected or not) is independent, each connected node pair would have a similar fraction of time when one node is infected whereas the other is susceptible, i.e. the period that allows epidemic to spread. In this case, the difference $\chi(T)$, where $0\leq \chi(T)<1$, in infection probability along a link within an arbitrary time $T$ may indicate the difference in the fraction of infected nodes between the homogeneous and heterogeneous SIS in the metastable state. Both the heterogeneous infection rates and the heterogeneous network topology contribute to the heterogeneity in the infection probability of each node. When the recovery rate is low or equivalently the epidemic prevalence is high, however, the infection probabilities of the nodes tend to be similar. Hence, $\chi(T)$ could suggest the difference in the fraction of infected nodes between the heterogeneous and homogeneous cases when the recovery rate is small. The larger the difference $\chi(T)$ is, the smaller the average fraction $y_\infty$ of infected nodes, in the metastable states of the heterogeneous SIS is. Equation (\ref{Equ1}), thus suggests that, the larger even-order moments of the infection rates lead to a smaller average fraction of infected nodes $y_\infty$, but the odd-order moments contribute in the opposite way. These theoretical results help us better understand our two observations in Fig.~\ref{fig:HSIS-1} and \ref{fig:HSIS-2}, when the recovery rates are small: (a) the average fraction $y_\infty$ of infected nodes decreases with the increased variance, and (b) given the same variance, the average fraction $y_\infty$ of infected nodes is lower if the third moment of the distribution is smaller.

\subsection{The log-normal distribution vs.\ the gamma distribution}
\label{sec:LNvsG}
To explore how fast $y_{\infty}$ decays, we perform simulations with different recovery rates $\delta$ and fit the curves of $y_\infty$ vs.\ the variance $v$. We find that, as shown in Fig.~\ref{fig:HSIS-3}, the relationship between the average fraction $y_{\infty}(v)$ of infected nodes and the variance $v$ can be fitted by a double-exponential function $y_{\infty, L}(v)=c_1e^{-c_2v}+c_3e^{-c_4v}$ and a quadratic function $y_{\infty, \Gamma}(v)=c_1v^2-c_2v+c_3$, when the infection rates follow log-normal and gamma distributions respectively. The coefficients $c_1,~c_2,~c_3$, and $c_4$, shown in Table \ref{tab0}, also suggest that, approximately, $y_{\infty, L}$ decreases exponentially with the variance $v$ much slower than the linear decrease of $y_{\infty, \Gamma}$ when $y_{\infty, \Gamma}$ is not close to $0$. 
\begin{figure}[htbp]
\centering
\includegraphics[scale=.3]{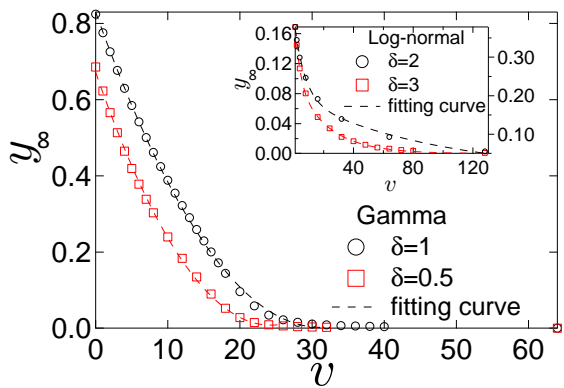}
\caption{The average fraction $y_\infty$ of infected nodes as a function of the variance $v$ of infection rates following gamma distributions. The recovery rates $\delta$ are different: $1$ ($\circ$) and $0.5$ ($\textcolor{red}{\square}$), and the dash lines are fitting curves. The simulations are on ER networks with average degree $\langle k\rangle=4$ and network size $N=10^4$. The results are averaged over $1000$ realizations. The inset contains the results about log-normal distributions.}
\label{fig:HSIS-3}
\end{figure} 

\begin{table}[!tbhp]
\caption{The coefficients of the fitting functions of $y_\infty$ vs. the variance $v$ for different infection-rate distributions under different recovery rates.}
\label{tab0}
\centering
\begin{tabular}{|c|c||c|c|c|c|}
\hline 
Dist. & $\delta$ & $c_1$ &  $c_2$ & $c_3$ & $c_4$\\
\hline
\multirow{2}{*}{$Log-\mathcal{N}$} & $3$ & $0.098$ & $0.045$ & $0.099$ & $0.28$\\
\cline{2-6} 
&$2$ & $0.20$ & $0.011$ & $0.21$ & $0.14$\\
\hline 
\multirow{2}{*}{$\Gamma$} & $1$ & $0.0011$ & $0.055$ & $0.67$ & \multirow{2}{*}{$N/A$} \\
\cline{2-5} 
& $0.5$ & $0.00085$ & $0.053$ & $
0.83$& \\ 
\hline
\end{tabular}
\end{table}

Besides the theoretical explanation as mentioned before, we explore further the physical interpretations of the difference in the fraction of infected nodes between the log-normal and gamma distributed infection rates. We define $r(\beta)$ as the ratio between the PDF of the log-normal and gamma distribution, \textit{i.e.}\ $r(\beta)=\frac{f_B(\beta;\mu, \sigma)}{f_B(\beta;k,\theta)}$. Thus $\lim\limits_{\beta\rightarrow0}r(\beta)=0$ and $\lim\limits_{\beta\rightarrow\infty}r(\beta)=\infty$. This reveals that if we set the same mean and variance (large) for both distributions, the log-normal distribution tends to generate a few extremely large values whereas the gamma distribution generates many extremely small values to produce the large variance. 

\begin{table}[!tbhp]
\renewcommand{\arraystretch}{1.3}
\caption{The percentiles of the log-normal and gamma distribution with the mean $m_1=1$ and variance $v=16$}
\label{tab1}
\centering

\begin{tabular}{|c||c|c|}
\hline 
Percentiles & $Log-\mathcal{N}$ & $\Gamma$\\
\hline
$1^{st}$ & $0.00483$ & $9.44\times10^{-32}$\\
\hline
$2.5^{th}$ & $0.00895$ & $2.20\times10^{-25}$\\
\hline
$5^{th}$ & $0.0152$ & $1.44\times10^{-20}$\\
\hline
$10^{th}$ & $0.0280$ &  $9.44\times10^{-16}$\\
\hline
$25^{th}$ & $0.0779$ & $2.20\times10^{-9}$\\
\hline
$50^{th}$ & $0.243$ & $1.44\times10^{-4}$\\
\hline
\end{tabular}
\end{table}

In Table~\ref{tab1}, we show the percentiles\footnote{A percentile is a measure to indicate the value below which a given percentage of observations in a group of observations fall.} of the two distributions with a large variance $v=16$. In a group of random numbers generated by the gamma distribution, $25\%$ of them are even smaller than $2.2\times10^{-9}$. The infection events driven by such small rates can hardly happen. However, in the infection rates generated by the log-normal distribution, even the first $1\%$ smallest values are large enough to make possible infections. Hence, the gamma distribution effectively filters the network more than the log-normal distribution, and reduce the spread of the epidemic more. This interpretation is also consistent with the theoretical explanation of the influence of the third moment of a distribution. The same large variance can be introduced by the log-normal distribution via the possibility of generating a large value and by the gamma distribution via the high probability of generating extremely small values. However, the gamma distribution leads to a smaller third moments compared to the log-normal distribution and the small infection rates it generates effectively filter the network, reducing the epidemic spread. 

\section{Large recovery rates}
We have shown that when the recovery rates are small, the i.i.d.~heterogeneous infection rates retards the epidemic spreading and the larger variance of infection rates leads to a smaller average fraction of infected nodes. Moreover, we further explained the influence of the higher moments of the infection rate on epidemic spreading. In this section, we discuss how the heterogeneous infection rates influence the epidemic spreading when the recovery rate is large, thus, the epidemic is close to die out. As an example, we show the simulation results of the SF networks with the log-normal distributed infection rates. We find that, the heterogeneous infection rates may increase the probability that the epidemic spreads out when the recovery rate is large, though if the epidemic can spread out, the larger variance of infection rates still leads to a smaller average fraction of infected nodes in the metastable state. 

\label{sec:SF}
We first employ the log-normal distribution for the heterogeneous infection rates and set the recovery rate $\delta=20$. As shown in Fig.~\ref{fig:HSIS-n1}, though the average fraction $y_\infty$ of infected nodes is close to $0$ (due to the large recovery rate), we can observe that the larger variance may lead to a slightly larger average fraction $y_\infty$ of infected nodes. However, the error bars (the standard deviation of the simulation results from different realizations) are large as compared to the average fraction of infected nodes. This is due to the fact that when the epidemic is close to die out on average, i.e.\ when $\delta=20$,  the epidemic dies out in some iterations of the simulations but spreads out with a nonzero fraction of infected nodes in the metastable state in the others. 
\begin{figure}[!t]
\centering
\includegraphics[scale=.35]{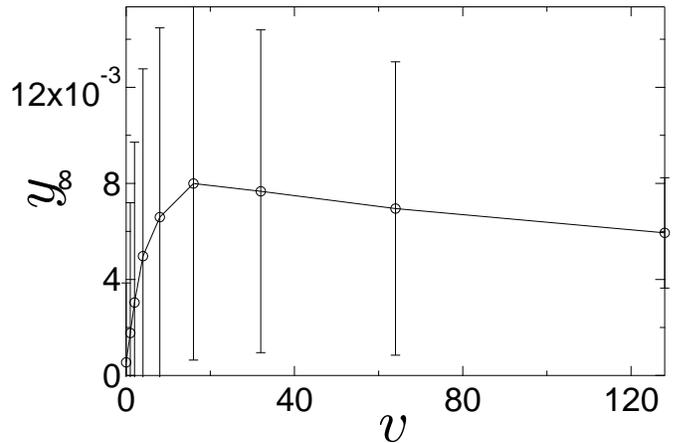}
\caption{The average fraction $y_{\infty}$ of infected nodes as a function of the variance $v$ of infection rates following the log-normal distribution, and the recovery rate $\delta=20$. The simulations are on SF networks with the exponent $\lambda=2.5$ and the network size $N=10^4$. The results are averaged over $1000$ realizations, and the error bars are the standard deviations of the results in different realizations. }
\label{fig:HSIS-n1}
\end{figure}        

Fig.~\ref{fig:HSIS-n2} shows the percentage $p^*$ ($\in [0,1]$) of the spread-out realizations in all realizations and the average fraction $y_{\infty}^*$ of infected nodes in these nonzero-infection realizations as a function of the variance of the infection rates. Here the the simulations are on SF networks with the size $N=10^4$ and the exponent $\lambda=2.5$. Clearly, the average fraction of infected nodes obtained by averaging that in all realizations is $y_\infty=p^*y_\infty^*$. We find that, in all nonzero-infection realizations, the average fraction $y_\infty^*$ of infected nodes still decreases as the variance of the infection rates increases. The average fraction $y_\infty$ of infected nodes obtained from all realizations may increase as the variance of the infection rates increases, because the percentage $p^*$ of nonzero-infection realizations increases when the variance of the infection rates is small and increases. Hence, the heterogeneous infection rates may enhance the probability that the epidemic spreads out. This can be explained as follows: the heterogeneous infection rates and the hubs in scale-free networks enable those links with a large infection rate to form a connected subgraph, allowing the epidemic to spread out. However, when the variance $v$ is large and further increases, as shown in Fig.~\ref{fig:HSIS-n2}, the fraction of non-zero infection realizations decreases. This is because, a large variance $v$ of the infection rates produces fewer large infection rates, prohibiting the formation of a connected subgraph with high infection rates that allows the epidemic to spread. However, the average fraction of infected nodes of the nonzero-infection realizations tend to decrease with the variance or heterogeneity of the infection rates. 
\begin{figure}[!t]
\centering
\includegraphics[scale=.32]{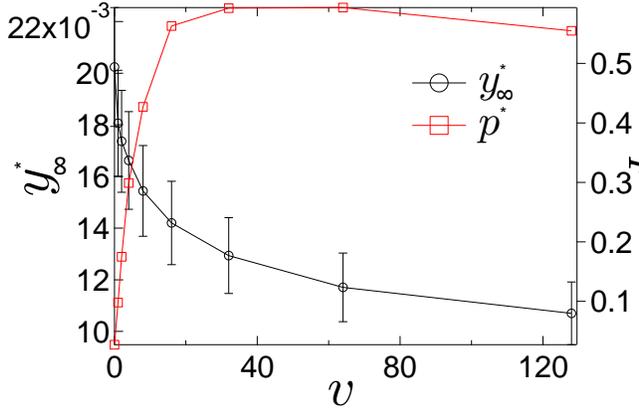}
\caption{The average fraction $y^*_{\infty}$ ($\circ$) of infected nodes in the nonzero-infection realizations and the percentage $p^*$ ($\textcolor{red}{\square}$) of the nonzero-infection realizations as a function of the variance $v$ of infection rates. The infection rates follow the log-normal distribution and the recovery rate $\delta=20$. The simulations are on SF networks with the exponent $\lambda=2.5$ and the network size $N=10^4$. The results are averaged over $1000$ realizations, and the error bars are the standard deviations of the results in different realizations. }
\label{fig:HSIS-n2}
\end{figure}   

If we increase the recovery rate to ensure that the epidemic dies out in the homogeneous case, \textit{i.e.}\ the effective infection rate is below the epidemic threshold $\tau_c$ in the classic SIS model, we obtain the same conclusions: the average fraction $y_\infty^*$ of infected nodes in nonzero-infection realizations (if exist)  always decreases as the variance of the infection rates increases, and the heterogeneous infection rates may increase the probability that the epidemic spreads out. 

We further compare the simulation results between the log-normal and gamma distributions. As shown in Fig.~\ref{fig:HSIS-n3a}, the average fraction $y^*_{\infty}$ of infected nodes in nonzero-infection realizations is larger when the infection rates follow the log-normal distribution than the gamma distribution. This observation is consistent with our previous observations and conclusions as illustrated in Section \ref{smallRecov}, when the variances of the infection rates are the same, the larger third moments of the infection rates lead to the more severe infection. However, as shown in Fig.~\ref{fig:HSIS-n3b}, when the variance of the infection rates is small, the percentage $p^*_\Gamma$ of the nonzero-infection realizations is larger in the case of the gamma distributed infection rates than the percentage $p^*_L$ of the nonzero-infection realizations in the case of the log-normal distributed infection rates. Moreover, as the variance of the infection rates is relatively large (for example, around $30$ in Fig.~\ref{fig:HSIS-n3b}) and increases, $p^*_\Gamma$ decreases faster than $p^*_L$, and $p^*_\Gamma$ could be smaller than $p^*_L$ if the variance is large enough. Given a network and a large recovery rate, more large infection rates lead to a higher probability that the epidemic can spread out. As in Section \ref{sec:LNvsG}, we can explain the observations in Fig.~\ref{fig:HSIS-n3b} by exploring the percentiles of the log-normal and gamma distributions with the mean $1$ in Table \ref{tabv16-128}, where two values ($16$ and $128$) of the variance are employed as examples. When the variance is $16$, there are more large values in a group of random numbers generated by the gamma distribution than the log-normal distribution; however, when the variance increases to $128$, though the first $1\%$ largest values of the gamma distribution are still larger than those of the log-normal distribution, there are more large values in the group of the log-normal random numbers. Hence, with the same small variance, the gamma distributed infection rates contribute more to the survival of the epidemic than the log-normal distributed infection rates, whereas with the same large variance, the log-normal distributed infection rates may lead to a higher probability that the epidemic spreads out.  
\begin{figure}[!thbp]
\centering
\subfigure[]{
\includegraphics[scale=.3]{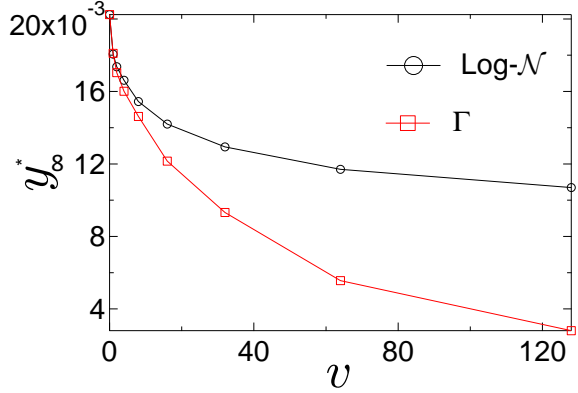}
\label{fig:HSIS-n3a}
}
\subfigure[]{
\includegraphics[scale=.3]{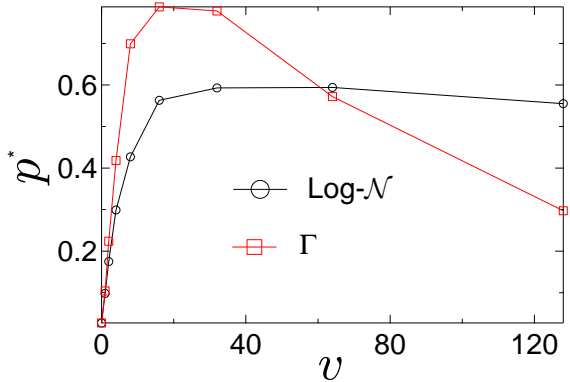}
\label{fig:HSIS-n3b}
}

\caption{(a) The average fraction $y^*_\infty$ of infected nodes in the nonzero-infection realizations and (b) The percentage $p^*$ of the nonzero-infection realizations in all realizations as a function of the variance $v$ of the infection rates which follow the gamma ($\circ$) and log-normal ($\textcolor{red}{\square}$) distribution.}
\end{figure} 

\begin{table}[!tbhp]
\renewcommand{\arraystretch}{1.3}
\caption{The percentiles of the log-normal and gamma distributions with the mean $m_1=1$ and variance $v=16$ and $128$}
\label{tabv16-128}
\centering

\begin{tabular}{|c||c|c|c|c|}
\hline 
 & $Log\mathrm{-}\mathcal{N}$ & $\Gamma$&$Log\mathrm{-}\mathcal{N}$ & $\Gamma$\\
 & $v=16$ & $v=16$& $v=128$ & $v=128$\\
\hline
$99^{th}$ & $12.1936$ & $19.9409$& $15.1178$ & $24.2306$\\
\hline
$98^{th}$ & $7.7225$ & $12.9981$& $8.2133$ & $5.7424$\\
\hline
$97^{th}$ & $5.7697$ & $9.3923$& $5.6176$ & $1.5509$\\
\hline
$96^{th}$ & $4.6209$ &  $7.1783$& $4.2063$ & $0.4140$\\
\hline
$95^{th}$ & $3.8751$ & $5.6325$& $3.3260$ & $0.1507$\\
\hline
\end{tabular}
\end{table}

We observe the same in ER networks as shown in the Appendix. Moreover, the links with i.i.d.~large infection rates are more likely to form a subgraph in SF networks than in ER networks, because of the existence of the nodes with large degrees in SF networks. Hence, with the similar value of the average fraction $y^*_\infty$ of infected nodes in the nonzero-infection realizations, we find that the percentage $p^*$ of the nonzero-infection realizations is much smaller in ER networks than SF networks.

We further consider an extreme case of SF networks -- the star network: one central node $n_0$ connects with all the other $m$ ($m\gg 1$) side nodes $n_i$ ($i=1,2,...,m$), and there is no link between any pair of the side nodes. By designing a specific distribution of the heterogeneous infection rates, we can always give a value of the recovery rate $\delta$ so that the epidemic spreads out with the heterogeneous infection rates but dies out with the corresponding homogeneous infection rates in a finite-size star network. In the classic model, the epidemic threshold of a star network is $\tau_c=\frac{\beta}{\delta}=\frac{1}{\sqrt{m}}$ \cite{van2011n}. If we set the homogeneous infection rate $\beta=1$ and the recovery rate $\delta=\sqrt{m}+\epsilon$, where $\epsilon$ is a positive but small constant number, then the epidemic dies out. With the same recovery rate, we set the heterogeneous infection rate with the distribution $Pr[B=2-\epsilon_1]=Pr[B=\epsilon_1]=0.5$, where $\epsilon_1$ is again a small and positive constant number, thus the average infection rate $E[B]=1$. We now look at the subgraph which is composed of the central node and approximately $\frac{m}{2}$ side nodes connected to the central node with infection rate $\beta_{sub}=2-\epsilon_1$. The effective infection rate is $\tau_{sub}=\frac{\beta_{sub}}{\delta}=\frac{2-\epsilon_1}{\sqrt{m}+\epsilon_1}\approx\frac{2}{\sqrt{m}}>\frac{1}{\sqrt{m/2}}\approx\tau_{c,sub}$, where $\tau_{c,sub}$ is the epidemic threshold of the subgraph. Hence, with the same recovery rate and the same average infection rates, the epidemic dies out in the homogeneous case but spreads out in the aforementioned heterogeneous case.

\section{Real-world networks}
\label{Chap:Real}
As mentioned in Section~\ref{Chap:Hbeta}, the interaction frequency between two nodes in a real-world network has been considered as the infection rate between the pair of nodes. In this section, we choose two real-world networks as examples to illustrate how their heterogeneous infection rates affect the spread of SIS epidemics on these networks. The heterogeneous infection rates from the datasets are normalized by the average so that the average is $1$. We compare the average fraction of infected nodes in the metastable state of the two networks in the 3 scenarios: 1) each network is equipped with its normalized original heterogeneous infection rates (hetero-$\beta$) as given in the dataset; 2) each network is equipped with the infection rates in the normalized original dataset but randomly shuffled (shuffled-$\beta$); 3) each network is equipped with a constant infection rate (homo-$\beta$) which equals to the average infection rate of the normalized original infection rates as given in the datasets. The heterogeneous infection rates in each network described in Scenario $1$ are possibly correlated. For example, the infection rate of a link may depend on the degrees of the two ending nodes of this link. The shuffling in Scenario 2 effectively removes the correlation if it exists, and the infection rates in Scenario 3 are homogeneous as in the classic SIS model. Our objective is to explore the relation between the infection rates and average fraction of infection in these 3 scenarios for both networks to verify our previous findings. 

The first network is the airline network where the nodes are the airports, the link between two nodes indicates that there's at least one flight between these two airports, and the infection rate along a link is the number of flights between the two airports. We construct this network and its infection rates from the dataset of openFlights\footnote{http://openflights.org/data.html}. The other one is the co-author network, where the nodes are the authors of papers, the link represents that the two corresponding authors have at least one collaborated paper, and the infection rate is the collaboration frequency\cite{newman2001structure}. 

Besides the infection rates, the network topology may as well influence the spread of SIS epidemics. We explore the most fundamental network feature of the two networks: the degree distributions which are shown in Fig.~\ref{fig:HSIS-4}. We can see that the degree distributions of the airline network and co-author network approximately follow a power law with the slope $\lambda=1.5$ and $2.5$ respectively. Hence, the degree distributions of the two networks influence the spread of epidemics in a similar way. More details of the two networks are listed in Table~\ref{tab2}. Note that we normalized the infection rates of each network by its mean so that the average rate is $1$.  

\begin{figure}[htbp]
\centering
\includegraphics[scale=.3]{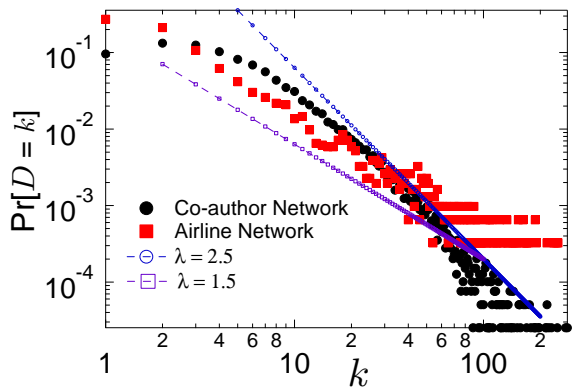}
\caption{The degree distributions of the airline ($\textcolor{red}{\blacksquare}$) and co-author network($\bullet$) can be approximately fitted by the power law distribution with a slope $\lambda=1.5$ and $2.5$ respectively.}
\label{fig:HSIS-4}
\end{figure} 

\begin{table}[!tbhp]
\renewcommand{\arraystretch}{1.3}
\caption{The number of nodes, number of links, variance of infection rates and range of infection rates in the two networks.}
\label{tab2}
\centering
\begin{tabular}{|c|c|c|c|c|}
\hline 
Name & Nodes &  Links & Variance & Range\\
\hline
Airline & $3071$ & $15358$ & $0.5560$ & $[0.2383, 11.0626]$\\
\hline
Co-author & $39577$ & $175692$ & $3.0566$ & $[0.0678,90.4625]$\\
\hline 
\end{tabular}
\end{table}

The distributions of the infection rates from the two networks are shown in Fig.~\ref{fig:51} and \ref{fig:52}. We find that, approximately the infection rates of the airline network are exponentially distributed, whereas those of the co-author network follow a log-normal distribution. Both of the two datasets support our previous choices of the infection-rate distribution.  

\begin{figure}[!thbp]
\centering
\subfigure[]{
\includegraphics[scale=.3]{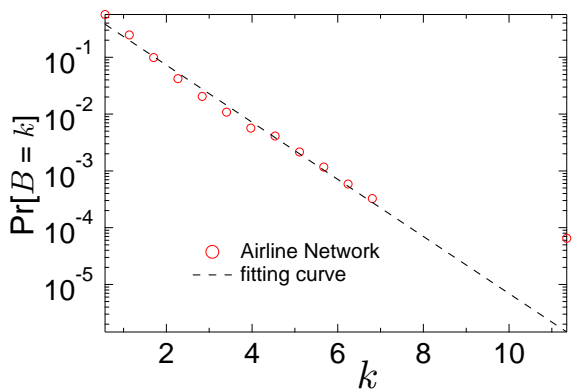}
\label{fig:51}
}
\subfigure[]{
\includegraphics[scale=.3]{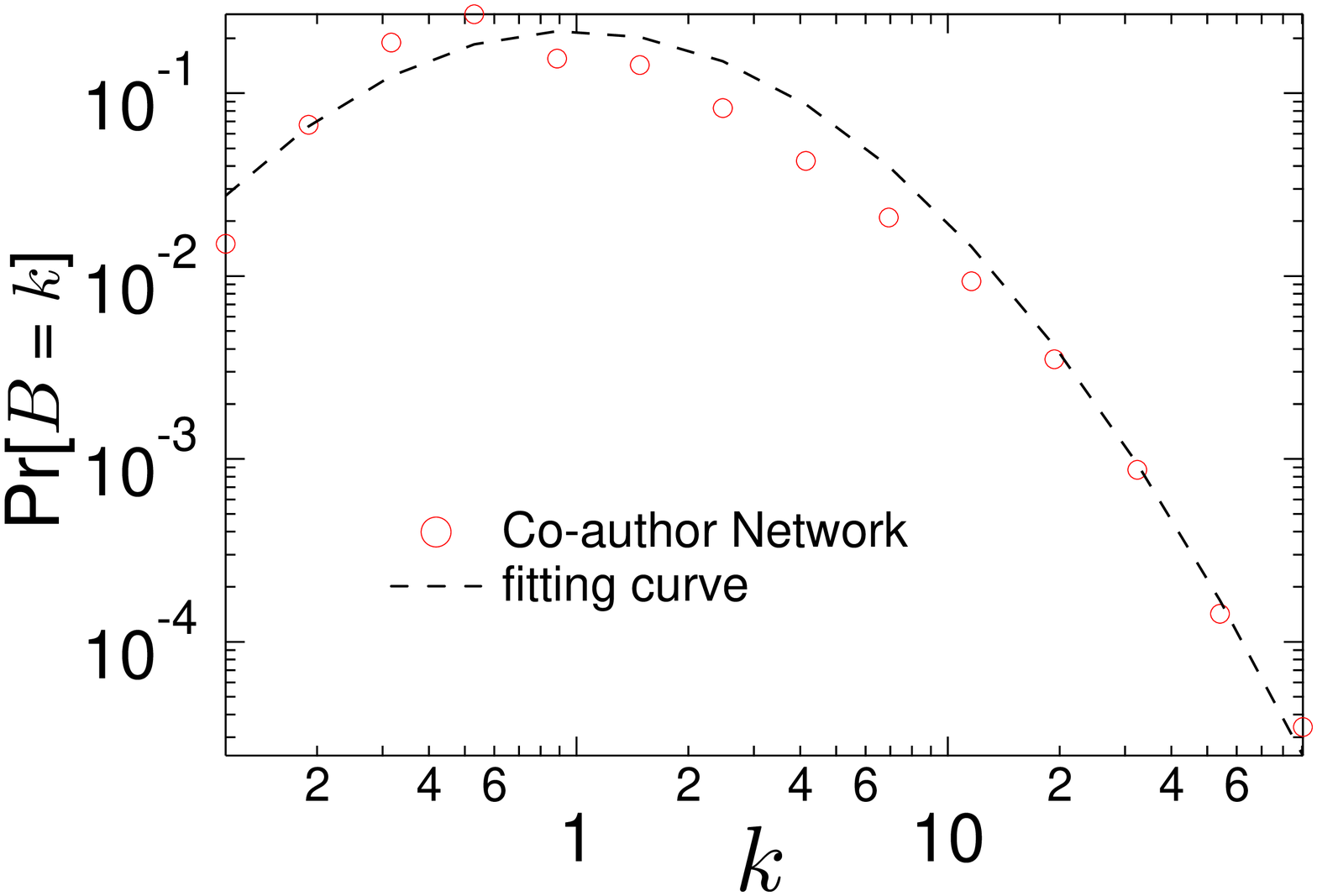}
\label{fig:52}
}

\caption{The distribution of the infection rates from real-world networks: (a) airline network and (b) co-author network. In each figure, the distribution ($\textcolor{red}{\circ}$) and fitting curve (dash line) are shown. The fitting curves are exponential and log-normal distributions in (a) and (b) respectively.}
\end{figure} 
\subsection{Small recovery rates}
We first consider the small recovery rates, with which the epidemic does not die out in any realizations. In this paper, we assume that the infection rates are i.i.d.~which corresponds to Scenario 2. As shown in Fig.~\ref{fig:61} and \ref{fig:62}, the average fraction $y_\infty$ of infected nodes in Scenario homo-$\beta$ is always larger than that in Scenario shuffled-$\beta$, which confirms our conclusion that the heterogeneity of infection rates on average retards the contagion processes of epidemics, when recovery rates are not very large. Moreover, we find that the reduction $y_{\infty,\text{homo-}\beta}-y_{\infty,\text{shuffled-}\beta}$ is larger in the co-author network, which has a larger variance of infection rates, than that in the airline network\footnote{We assume that the two networks have a similar topology, since they have a similar degree distribution as shown in Fig.~\ref{fig:HSIS-4}}. This observation verifies our conclusion that, the larger the variance of the infection rates is, the smaller $y_\infty$ is. Compared to the independent infection rates in the case shuffled-$\beta$, the possibly correlated infection rates in the case hetero-$\beta$ can further decrease (in e.g.~the airline network) or increase (in e.g.~the co-author network) the average fraction of infected nodes. This observation points out a new challenging question: what is the influence of such correlated heterogeneous infection rates on the SIS epidemics.       

\begin{figure}[!thbp]
\centering
\subfigure[]{
\includegraphics[scale=.3]{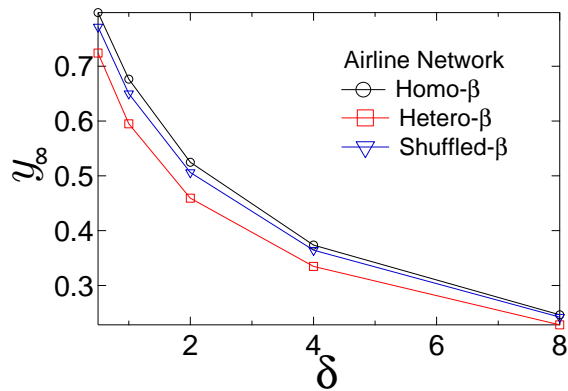}
\label{fig:61}
}
\subfigure[]{
\includegraphics[scale=.3]{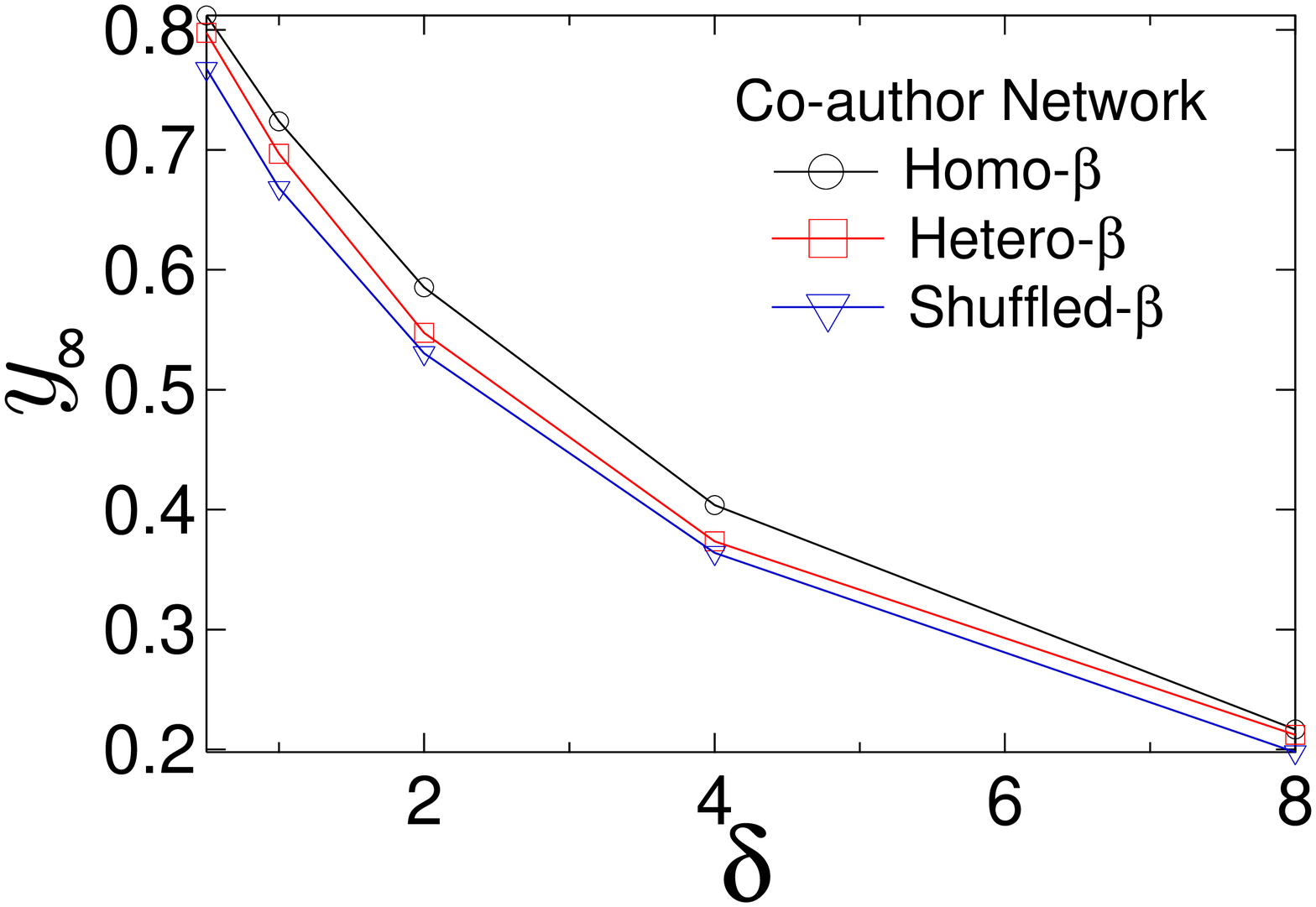}
\label{fig:62}
}
\caption{The average fraction $y_\infty$ of infected nodes as a function of the recovery rate $\delta$. The networks are from real-world: (a) airline network and (b) co-author network. In each figure, the SIS model with homogeneous ($\circ$), original heterogeneous ($\textcolor{red}{\square}$) and shuffled heterogeneous ($\textcolor{blue}{\triangledown}$) infection rates are compared.}
\end{figure}

\subsection{Large recovery rates}
As shown in Fig.~\ref{fig:7}, when the recovery rate increases and the effective infection rate is close to the epidemic threshold, the average fraction $y_\infty$ of infected nodes in the Scenario hetero-$\beta$ becomes mostly larger than that in the other two scenarios. Besides that, it is still consistent with our previous conclusion that if $y_{\infty,\text{homo-}\beta}\ne 0$, then $y_{\infty,\text{homo-}\beta}>y_{\infty,\text{shuffled-}\beta}$. Moreover, in the co-author network, we observe that when the recovery rate $\delta=40$, $y_{\infty,\text{shuffled-}\beta}>y_{\infty,\text{homo-}\beta}=0$. However, in the airline network, we cannot observe that $y_{\infty,\text{shuffled-}\beta}>y_{\infty,\text{homo-}\beta}$ with any selected recovery rate, and this may be because of the small variance of the infection rates. Hence, the observations verify our conclusions that if the epidemic spreads out with the homogeneous infection rates, then the overall infection is always more severe than that with the heterogeneous infection rates (i.i.d.~and with the same mean as the homogeneous infection rate); however, the heterogeneous infection rate may contribute to the survival of the epidemic. 
\begin{figure}[!thbp]
\centering
\subfigure[]{
\includegraphics[scale=.3]{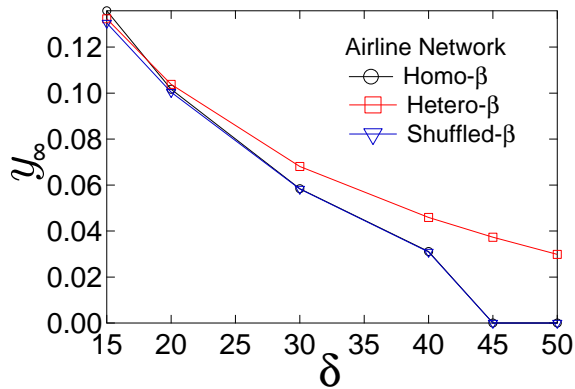}
\label{fig:71}
}
\subfigure[]{
\includegraphics[scale=.3]{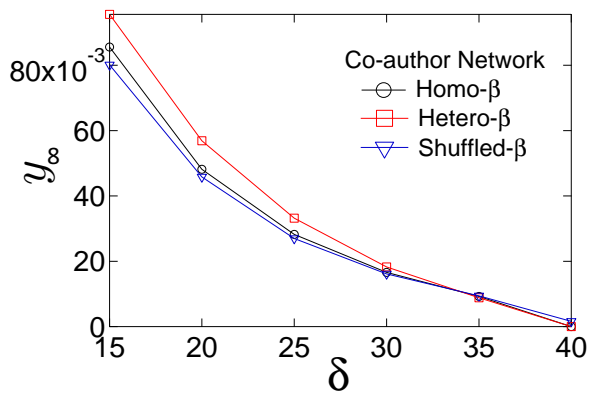}
\label{fig:72}
}
\caption{The average fraction $y_\infty$ of infected nodes as a function of the recovery rate $\delta$. The recovery rate is very large so that the effective infection rate is close to the epidemic threshold. The networks are from real-world: (a) airline network and (b) co-author network. In each figure, the SIS model with homogeneous ($\circ$), original heterogeneous ($\textcolor{red}{\square}$) and shuffled heterogeneous ($\textcolor{blue}{\triangledown}$) infection rates are compared.}
\label{fig:7}
\end{figure}

\section{Discussions}
In summary, we illustrate with simulations, theoretical analysis and physical interpretations that, when the recovery rate is small, the heterogeneity of infection rates on average retards the virus spread and whereas the larger even-order moments of the infection rates tend to lead to a smaller $y_\infty$, the odd-order moments contribute in the other way around; when the recovery rate is large so that the epidemic may die out, the heterogeneous infection rates may enhance the probability that the epidemic spread out. We also verify the influence of the heterogeneity of infection rates on virus spread in real-world networks. Our work reveals that the higher moments, especially the variance, of the infection rates may evidently affect the epidemic spread, even far more seriously than intuitively expected. Our finding implies that real-world heterogeneous epidemic spread may not be as severe as the classic homogeneous SIS model predicts, but the heterogeneous epidemic may not be as easy as the homogeneous SIS model indicates to die out. 

In this work, we have focused on the Markovian SIS where the time for an infected node $i$ to infect a susceptible neighbor $j$ is an exponential random variable with rate $\beta_{ij}$. Theorem 1 can be extended to Non-Markovian SIS models with heterogeneous infection rates where the infection time between a neighboring infected susceptible node pair $(i,j)$ with average $1/\beta_{ij}$ follows a distribution other than the exponential distribution. Such extension to Non-Markovian SIS models is possible if $1-F(\tau;\beta)$ the probability that the infection time is larger than $\tau$ when the average infection time is $1/\beta$ is a convex function of $\beta$.

The time for an infected node to infect a susceptible neighbor is more in depth and detailed information. Infection time measurement becomes possible though in general is still challenging. For example, in the experiments of the epidemic in the plant population, the infection time can be measured. As more such datasets become available, it would be interesting to tackle a new direction: what is the influence of the heterogeneous infection time on viral spreading?

\section*{Acknowledgements}
\thispagestyle{empty}

We thank Shlomo Havlin and Piet Van Mieghem for their inspiring comments regarding to the problem definition and the mathematical proofs respectively. We also wish to thank CONGAS (Grant No.\ FP7-ICT-2011-8-317672) for support.

\newpage
\begin{appendix}
\begin{figure}[!thbp]
\label{fig:A1}
\centering
\subfigure[]{
\includegraphics[scale=.3]{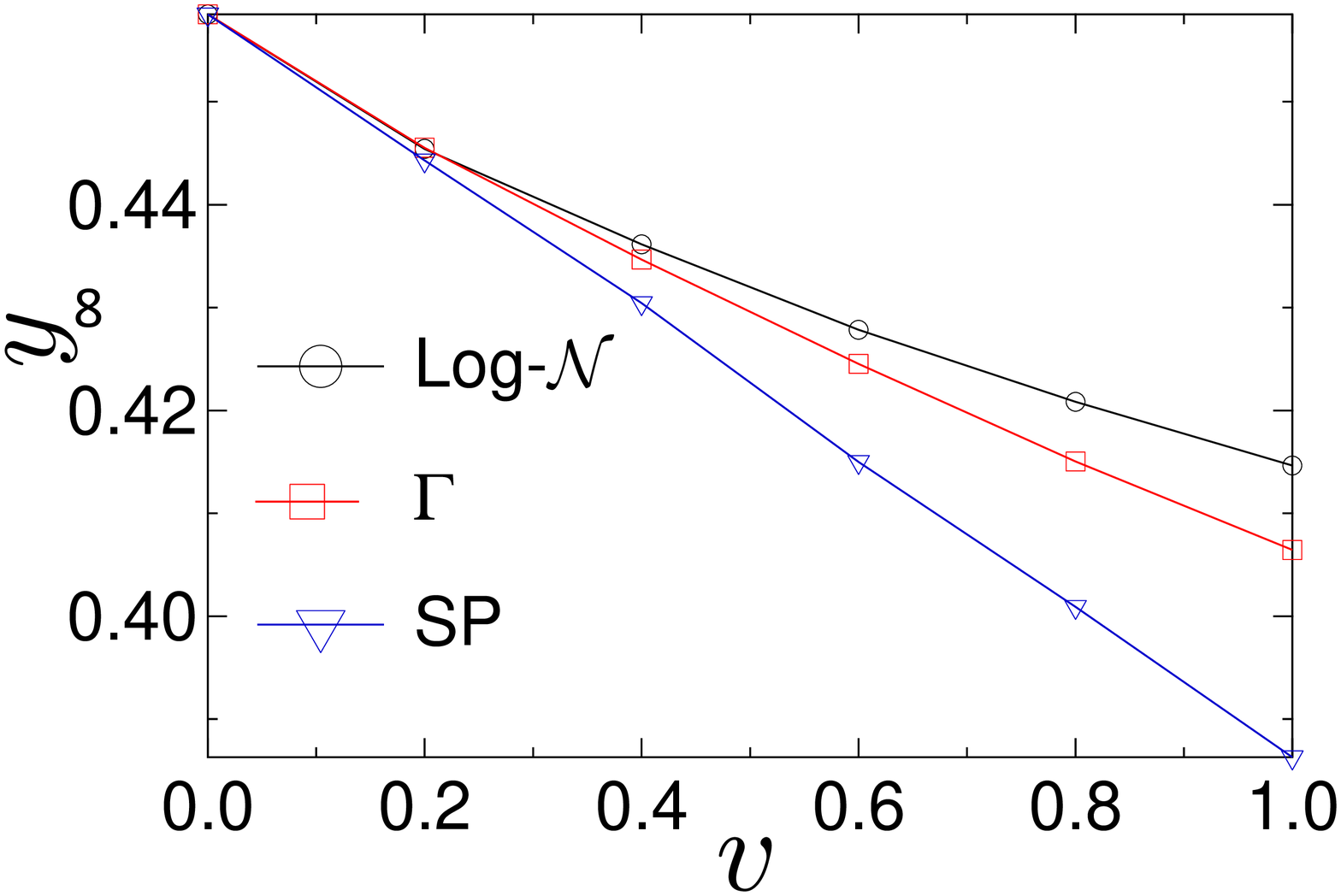}
\label{fig:a01}
}
\subfigure[]{
\includegraphics[scale=.3]{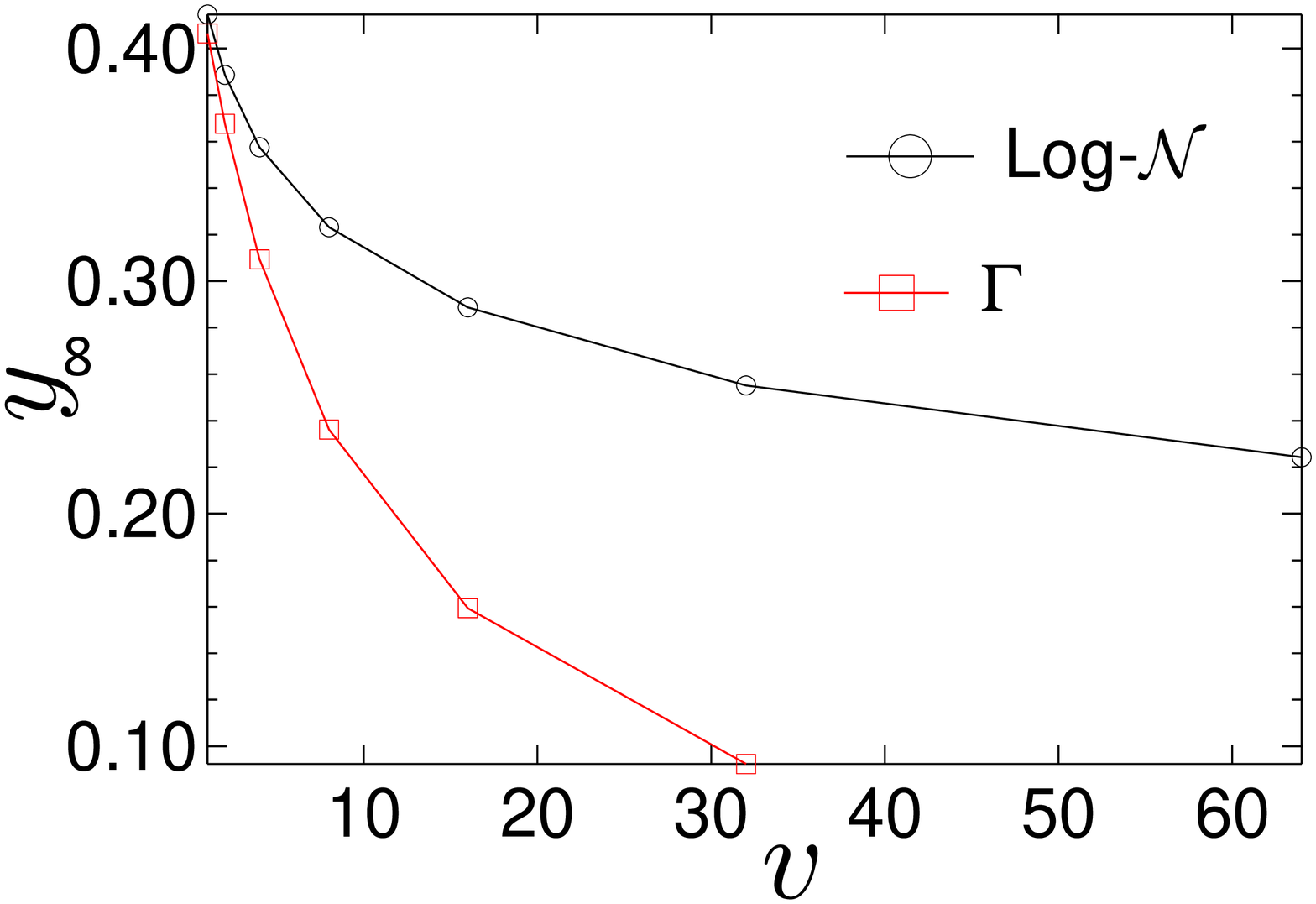}
\label{fig:a02}
}

\caption{The average fraction $y_\infty$ of infected nodes as a function of the normalized variance $v$ ((a) $v\leq 1$ and (b) $v\geq 1$) of infection rates for log-normal ($\circ$), gamma ($\textcolor{red}{\square}$), and SP (($\textcolor{IgorBlue}{\triangledown}$) infection rates distribution respectively, and the recovery rate $\delta=2$. We consider SF networks with the exponent $\lambda=2.5$ and network size $N=10^4$. The results are averaged over $1000$ realizations.}
\end{figure}

\begin{figure}[!thbp]
\label{fig:A2}
\centering
\subfigure[]{
\includegraphics[scale=.3]{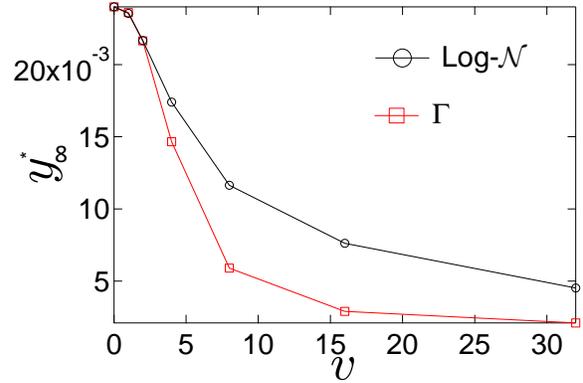}
\label{fig:a03}
}
\subfigure[]{
\includegraphics[scale=.3]{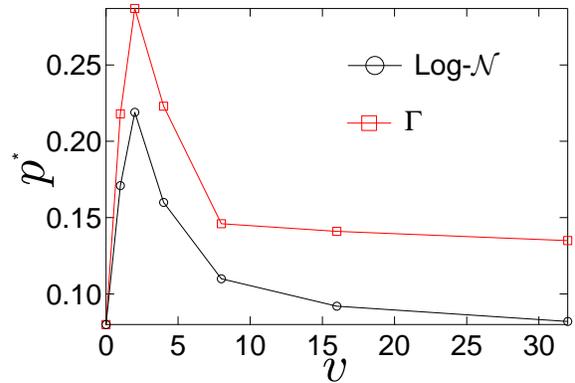}
\label{fig:a04}
}

\caption{(a) The average fraction $y^*_\infty$ of infected nodes in the nonzero-infection realizations and (b) The percentage $p^*$ of the nonzero-infection realizations in all realizations as a function of the variance $v$ of the infection rates which follow the gamma ($\circ$) and log-normal ($\textcolor{red}{\square}$) distribution. We consider ER networks with the average degree $E[D]=4$ and network size $N=10^4$. The recovery rate is $\delta=3.95$. The results are averaged over $1000$ realizations.}
\end{figure}
\end{appendix}

\end{document}